\documentclass[%
reprint,
prxquantum,
superscriptaddress,
nofootinbib,
amsmath,amssymb,
aps,
]{revtex4-2}

\usepackage[colorlinks=true,linkcolor=blue,citecolor=blue,urlcolor=blue]{hyperref}
\usepackage{physics}
\usepackage{amsthm}
\usepackage{enumitem}
\usepackage[legacycolonsymbols]{mathtools}

\newtheorem*{thm*}{Theorem}
\newtheorem{cor}{Corollary}
\newtheorem*{cor*}{Corollary}
\newtheorem{dfn}{Definition}
\newtheorem*{dfn*}{Definition}

\begin{document}

\title{Quantum State Recovery via Direct Sum Formalism \\Without Measurement Outcomes}

\author{Taiga Suzuki}%
\affiliation{%
Department of Physics, Institute of Science Tokyo, Meguro-ku, Tokyo 152-9551, Japan
}%
\author{Masayuki Ohzeki}%
    \affiliation{%
    Department of Physics, Institute of Science Tokyo, Meguro-ku, Tokyo 152-9551, Japan
    }%
    \affiliation{%
    Graduate School of Information Sciences, Tohoku University, Sendai, Miyagi 980-9564, Japan
    }%
    \affiliation{%
    REISI, Kumamoto University, Kumamoto, Kumamoto 860-8555, Japan
    }%
    \affiliation{%
    Sigma-i Co., Ltd., Minato-ku, Tokyo 108-0075, Japan}%

\date{\today}%

\begin{abstract}
    This study proposes a new approach to quantum state recovery following measurement.
   Specifically, we introduce a special operation that transfers the probability amplitude of the quantum state into its orthogonal complement. This operation is followed by a measurement performed on this orthogonal subspace, enabling the undisturbed original quantum state to be regained.
   Remarkably, this recovery is achieved without dependence of the post-measurement operation on the measurement outcome, thus allowing the recovery without historical dependence.
   This constitutes a highly nontrivial phenomenon. From the operational perspective, as the no-cloning theorem forbids perfect and probabilistic cloning of arbitrary quantum states, and traditional post-measurement reversal methods typically rely on operations contingent on the measurement outcomes, it questions fundamental assumptions regarding the necessity of historic dependence.
   From an informational perspective, since this recovery method erases the information about the measurement outcome, it's intriguing that the information can be erased without accessing the measurement outcome.
   These results imply the operational and informational non-triviality formulated in a direct-sum Hilbert space framework.
\end{abstract}

\maketitle
\section{Introduction}
The state is inevitably disturbed when a measurement is performed on a quantum state. This raises a fundamental question: is restoring the disturbed state to its original form possible? 
The reversibility of quantum states has been the subject of extensive debate over the years, 
encompassing theoretical areas such as quantum foundations \cite{Zurek_2018,PhysRevA.83.044101,dass2017repeatedcontinuousweakmeasurements,Lee2021completeinformation,PhysRevA.93.032134}, 
the information-disturbance trade-off \cite{PhysRevLett.100.210504,PhysRevLett.112.050401,PhysRevLett.109.150402}, 
quantum error correction \cite{PhysRevA.54.1098, PhysRevLett.82.2598, Wang_2024} 
and even protocols such as quantum teleportation \cite{PhysRevLett.70.1895,PhysRevResearch.3.033119}. 
Beyond theory, the reversibility question has also been addressed in experiments, 
which demonstrate that partial recovery of post-measurement states can indeed be achieved 
in a variety of physical platforms \cite{Katz2008,Kim2012,PhysRevX.4.021043,XuYou-Yang_2010,Wang:2025uff,PhysRevLett.128.050401,PhysRevLett.97.166805}.
Therefore, understanding the mechanism of state recovery after measurement is a fundamental issue in quantum mechanics, with significant implications for theoretical, experimental, and applied quantum information science.

There are two major approaches to addressing this question.
The first is a well-studied method called Quantum Reversible Measurement (QRM), Measurement Reversal or Undoing Measurement~\cite{PhysRevLett.73.913,Terashima_2006,PhysRevA.53.3808,Jordan_2010}. For simplicity, we call these methods QRM. The disturbance caused by a quantum measurement can be described using measurement operators~\cite{nielsen00,Hayashi2015}. If this operator admits an inverse, it becomes possible—though only probabilistically—to recover the original state by applying this inverse operation after the measurement. Under these conditions, any quantum state can be measured and subsequently restored using QRM.

The second approach considers recovering the original quantum state by creating a copy of the state in advance, performing the measurement on the copied state, and retrieving the untouched original state. However, due to the no-cloning theorem~\cite{Wootters1982Single,Hayashi2015}, which prohibits perfect or probabilistic cloning of arbitrary quantum states, this approach is generally regarded as impossible, despite the existence of specific probabilistic cloning protocols that apply under restrictive conditions~\cite{PhysRevLett.80.4999}.

If, however, it were possible to realize this second approach through a suitable operation, it would constitute a highly counterintuitive and nontrivial phenomenon from both operational and informational perspectives. As for the operational perspectives, in contrast to conventional QRM, where the post-measurement operation necessarily depends on the measurement outcome due to the uniqueness of the inverse operation, this method would enable recovery of the original quantum state entirely independent of the measurement outcome.
As for the informational perspectives, the important feature of the recovery process is that the information gained by the first measurement is erased by the second recovery operation~\cite{PhysRevA.83.044101,PhysRevLett.74.1040,buscemi2017information,Jordan_2010}. Since the second approach enables us to recover the initial state independent of the measurement outcomes, this implies that the information erasure of the first measurement is done without accessing the measurement outcomes.
Thus, the existence of such a recovery mechanism would be both operationally and informationally intriguing.

Of particular interest is that this study succeeded in constructing a process that realizes the second approach by leveraging a mechanism. In this framework, the probability amplitude of a quantum state is not fully duplicated, but instead transferred to its orthogonal complement within a direct-sum Hilbert space decomposition.
This process enables the measurement of only the quasi-copied state and the probabilistic recovery of the original quantum state without any reliance on the measurement outcomes. Notably, the method applies to arbitrary quantum states and can be implemented physically.

Although various cloning methods have been proposed and analyzed in the past~\cite{RevModPhys.77.1225, PhysRevA.54.1844, PhysRevLett.80.4999,PhysRevA.62.012302,Fan_2003,Cerf_2000}, the idea of transferring the amplitude to an orthogonal complement within a direct-sum structure has not been explored. Thus, our result provides a novel example of how expressing quantum states in a direct-sum Hilbert space enables nontrivial phenomena, enriching the conceptual framework of quantum information theory.

\section{Preliminaries}
\subsection{Definition and Notations}
Firstly, let us introduce the definitions and notations used in this paper.
$d$-dimensional Hilbert spaces are described as 
$\mathcal{H}_d$. 
In this paper, only finite-dimensional Hilbert spaces are considered.
The orthogonal complement space of $\mathcal{H}_d$ is described as 
$\mathcal{H}_d^{\bot}$.
Additionally, the set of linear operators on the Hilbert space $\mathcal{H}_d$ is described as 
$L(\mathcal{H}_d)$ and the set of quantum states is described as $S(\mathcal{H}_d)$.

The following defines the operations that simultaneously perform linear operators on the following Hilbert space and its orthogonal complement space:

\begin{dfn} \label{dfn:direct sum}
  $\hat{A}$ and $\hat{B}$ are linear operators that act on a Hilbert space $\mathcal{H}_d$ and its orthogonal complement space $\mathcal{H}_d^{\bot}$, respectively.
  The followings are linear operators that act simultaneously on the respective subspaces:
  \begin{align}
    \begin{split}
    \hat{A}  \oplus  \hat{B} \in  L(\mathcal{H}_d\oplus \mathcal{H}_d^{\bot})  \stackrel{\mathrm{def.}}{\iff} 
    \begin{cases}
      \hat{A} : \mathcal{H}_d\to \mathcal{H}_d,\\
      \hat{B} : \mathcal{H}_d^{\bot}\to \mathcal{H}_d^{\bot},
    \end{cases}\\
    \hat{C}  \boxplus  \hat{D} \in  L(\mathcal{H}_d\oplus \mathcal{H}_d^{\bot})  \stackrel{\mathrm{def.}}{\iff} 
    \begin{cases}
      \hat{C} : \mathcal{H}_d^{\bot}\to \mathcal{H}_d,\\
      \hat{D} : \mathcal{H}_d\to\mathcal{H}_d^{\bot}.
    \end{cases}.
  \end{split}
  \end{align}
  Note that the linear operation $\hat{A}  \oplus  \hat{B}$ corresponds to the block diagonal terms, while $\hat{C}  \boxplus  \hat{D}$ corresponds to the block off-diagonal terms in the matrix representation.
\end{dfn}
The properties and theorems of the direct sum space to be used are summarized in Appendix \ref{app:theorems}.
\subsection{Framework of Quantum Measurement Theory}

In the standard framework of quantum measurement theory, 
a measurement on a quantum system described by a density operator 
$\hat{\rho} \in S(\mathcal{H}_d)$ 
is represented by a collection of linear operators 
$\{ \hat{M}_x \}_{x}$, labeled by one of the possible measurement outcomes $x\in \{x \}$. 
These operators satisfy the completeness relation
\begin{align}
\sum_x \hat{M}_x^\dagger \hat{M}_x = \hat{1},
\end{align}
where $\hat{1}$ denotes the identity operator on $\mathcal{H}_d$. 
The above condition ensures that the total probability of obtaining some outcome is normalized to unity.  

When a specific outcome $x$ is observed, the post-measurement state of the system is given by the transformation
\begin{align}
\hat{\rho} \;\mapsto\; \frac{\hat{M}_x \, \hat{\rho} \, \hat{M}_x^\dagger}{\mathrm{Tr}\!\left[\hat{M}_x \hat{\rho} \hat{M}_x^\dagger\right]},
\end{align}
with the associated probability
\begin{align}
P[x] = \mathrm{Tr}\!\left[\hat{M}_x \hat{\rho} \hat{M}_x^\dagger\right].
\end{align}

Importantly, every measurement process of this form can be physically realized 
by embedding the system into a larger Hilbert space, 
applying a global unitary evolution, 
and subsequently performing a projective measurement on an ancillary subsystem~\cite{nielsen00,Hayashi2015}.

\subsection{Quantum Reversible Measurement}\label{sec:QRM}

Quantum reversible measurement (QRM) provides a theoretical framework in which 
the disturbed quantum state can be restored to its pre-measurement form by means 
of a subsequent recovery measurement~\cite{PhysRevLett.73.913,Terashima_2006,PhysRevA.53.3808,Jordan_2010}. 

In this framework, a necessary condition for reversibility is that each of the 
initial measurement operators $\hat{M}_x$ is regular, i.e., admits an inverse 
operator $\hat{M}_x^{-1}$~\cite{PhysRevA.53.3808,PhysRevLett.82.2598}. 
Given outcome $x$, the recovery measurement can then be constructed by including, 
as one of its Kraus operators, an operator proportional to the inverse,
\begin{equation}
\hat{R}_x \propto \hat{M}_x^{-1}.
\end{equation}
The proportionality constant is chosen to ensure that the recovery operators 
satisfy the completeness relation. Because the inverse operator $\hat{M}_x^{-1}$ 
is unique, the construction of the recovery measurement inevitably depends on the 
specific outcome $x$.

Nevertheless, as will be discussed in the following section, the utilization of 
direct-sum structures in the Hilbert space makes it possible to achieve recovery 
of the original quantum state without explicit reliance on the measurement outcomes.

\section{Methodology}
Next, let us propose the method for copying an arbitrary quantum state to its orthogonal complement space, followed by the measurement and recovery processes. The procedure consists of the following four steps:
\begin{enumerate}[label=(\Alph*)]
  \item Preparation of an arbitrary unknown initial state and addition of a qubit ancillary system. \label{step:preparation}
  \item Implementation of a specialized quantum channel to create a copy of the quantum state in the orthogonal complement space. \label{step:copying}
  \item Execution of a measurement on one of the copied quantum states. \label{step:measurement}
  \item Recovery of the unmeasured quantum state. \label{step:recovery}
\end{enumerate}
Note that the final operation (D) is stochastic due to the imperfect nature of copying into the orthogonal complement space.

The subsequent subsections provide a detailed description of each operation.

\indent \ref{step:preparation} Preparation of the Initial System and Addition of the Ancillary System:\ Initially, consider a quantum state $\hat{\rho} \in S(\mathcal{H}'_d)$ defined on a $d$-dimensional Hilbert space $\mathcal{H}'_d$. 
An ancillary qubit system is then introduced, prepared in a fixed pure state $\ket{a} \in \mathcal{H}_2$, 
yielding the composite system in the tensor product space $\mathcal{H}'_d \otimes \mathcal{H}_2$. 
The resulting total state is given by
\begin{align}
\hat{\rho} \mapsto \hat{\rho} \otimes \ketbra{a} \in S(\mathcal{H}'_d \otimes \mathcal{H}_2).
\end{align}

Since there always exists a state $\ket{a^\bot}$ orthogonal to $\ket{a}$, 
the total Hilbert space admits the orthogonal decomposition
\begin{align}
\mathcal{H}'_d \otimes \mathcal{H}_2 = \mathcal{H}_d \oplus \mathcal{H}_d^\bot,
\end{align}
where
\begin{align}
  \mathcal{H}_d := \mathcal{H}'_d \otimes \mathrm{span}\{\ket{a}\},\quad
  \mathcal{H}_d^\bot := \mathcal{H}'_d \otimes \mathrm{span}\{\ket{a^\bot}\}. \label{def:h_d^bot}
\end{align}

With respect to this decomposition, the block-diagonal components of 
$\mathcal{H}_d$ and $\mathcal{H}_d^{\bot}$ can be represented in the orthonormal 
basis $\{\ket{\psi_i}\}$ of $\mathcal{H}'_d$ as
\begin{align}
  \begin{split}
    &(\bra{\psi_i}\otimes \bra{a})(\hat{\rho}\otimes\ketbra{a})(\ket{\psi_j}\otimes \ket{a}) 
      = \bra{\psi_i}\hat{\rho}\ket{\psi_j},\\
    &(\bra{\psi_i}\otimes \bra{a^{\bot}})(\hat{\rho}\otimes\ketbra{a})(\ket{\psi_j}\otimes \ket{a^{\bot}}) 
      = 0.\label{def:matrix_components}
  \end{split}
\end{align}
In a similar manner, evaluation of the block off-diagonal components shows $\hat{0}\boxplus \hat{0}$.
Therefore, the total state $\hat{\rho}\otimes\ketbra{a}$ on $\mathcal{H}'_d \otimes \mathcal{H}_2$ can be written, omitting the vanishing off-diagonal terms, as
\begin{align}
  \hat{\rho} \oplus \hat{0} \quad \text{on} \quad \mathcal{H}_d \oplus \mathcal{H}_d^{\bot}.\label{def:initial_state}
\end{align}

Notably, the matrix representation of $\hat{\rho} \oplus \hat{0}$ is independent of the specific choice of $\ket{a}$ when expressed as in Eq.\eqref{def:matrix_components}. Consequently, the decomposition is valid for any orthonormal basis of the ancillary qubit, resulting in the same quantum state on the direct-sum structure. Since operator properties on direct-sum spaces (Appendix~\ref{app:theorems}) are likewise basis-independent, the generality of the proposed method is preserved, though writing out the explicit forms of $\mathcal{H}_d$ and $\mathcal{H}_d^\bot$ is instructive for clarity. A particular choice of $\ket{a}$ with a more transparent physical interpretation is presented in Appendix~\ref{app:special_case}.

\indent \ref{step:copying} Quantum Channel for Copying the Quantum State to the Orthogonal Complement Space:\ The copying process can be implemented by applying a quantum channel which has  $\hat{K}_0,\hat{K}_1$ as the Kraus operators parameterized by an angle $\phi\in \mathbb{R}$:
  \begin{align}
    \begin{split}
    &\hat{K}_0 = (\cos\phi \cdot \hat{1})\oplus\hat{0} + (\sin\phi \cdot \hat{1})\boxplus\hat{0},\\
    &\hat{K}_1 = \hat{0}\oplus(-\cos\phi \cdot \hat{1}) + \hat{0}\boxplus(\sin\phi \cdot \hat{1}). \label{K_0,K_1}
  \end{split}
  \end{align}
The explanation of how these operators satisfy the conditions of Kraus operators is described in Appendix~\ref{app:validity}.

And then, this quantum channel is represented and calculated as follows using Corollary \ref{Product of operators on the direct sum space} and Corollary \ref{Sum of operators on the direct sum space} in Appendix~\ref{app:theorems}:
\begin{align}
  \mathcal{E}(\hat{\rho}\oplus\hat{0})= (\cos^2\phi \cdot \hat{\rho}) \oplus (\sin^2\phi \cdot \hat{\rho}).
\end{align}
Therefore, this state can be interpreted as the state that has been copied to the orthogonal complement space $\mathcal{H}_d^\bot$.
We call this state a quasi-copied state as it is not the perfect copy of the quantum state.

\indent \ref{step:measurement} Measurement on the quasi-copied state:\ 
Next, a measurement is performed on one of the quasi-copied states. The measurement process $\mathcal{M}$ is characterized as the measurement operators which combine the identity operator on $\mathcal{H}_d$ and the measurement operator $\hat{M}_{\nu}$ defined on $\mathcal{H}_d^\bot$ and which satisfies the conditions for a valid measurement operator as detailed in Appendix~\ref{app:validity}. This construction ensures that the subspace $\mathcal{H}_d$ remains undisturbed by the measurement, and the measurement is conducted only on the orthogonal complement space $\mathcal{H}_d^\bot$. 

Formally, the measurement process is expressed as:
\begin{align}
  \mathcal{M}:\left\{\left(\frac{1}{\sqrt{n}}\hat{1}\right)\oplus \hat{M}_{\nu}\right\}_{\nu},\quad \nu= 1,2,\ldots,n. \label{measurement}
\end{align}

By this measurement process, the probability of obtaining the measurement outcome $\nu$ is calculated as:
\begin{align}
  P[\nu] = \frac{\cos^2\phi}{n} + \sin^2\phi\cdot\text{tr}\left( \hat{M}_{\nu}\hat{\rho}\hat{M}_{\nu}^\dagger \right). \label{P[nu]}
\end{align}
On the other hand, the post-measurement state corresponding to the outcome $\nu$ is calculated as:
\begin{align}
  \frac{1}{P[\nu]} \left[ \left(\frac{\cos^2\phi}{n} \hat{\rho} \right)\oplus \left( \sin^2\phi\cdot \hat{M}_{\nu}\hat{\rho}\hat{M}_{\nu}^\dagger \right) \right].
\end{align}
As shown above, the measurement process is conducted only on the orthogonal complement space $\mathcal{H}_d^\bot$, and the quantum state in the subspace $\mathcal{H}_d$ remains unaffected by the measurement.

\indent \ref{step:recovery} Recovery of the initial quantum state:\ Finally, the initial quantum state, unaffected by the measurement operator, is recovered. This restoration is achieved through a subsequent measurement process $\mathcal{R}$, defined as:
\begin{align}
  \mathcal{R}:\{\hat{R}_{\mu_0}:=\hat{1} \oplus \hat{0},\> \hat{R}_{\mu_1}:=\hat{0} \oplus \hat{1}\}. \label{recovery}
\end{align}
The explanation of how these operators satisfy the conditions of Kraus operators is described in Appendix~\ref{app:validity}.
As we will see below, by post-selecting the outcome $\mu_0$, the restoration of the initial quantum state is successfully done.
The post-measurement state corresponding to the outcome $\mu_0$ is calculated as:
\begin{widetext}
  \begin{align}
    \frac{1}{P[\mu_0|\nu]} \left(\hat{1} \oplus \hat{0}\right) \frac{1}{P[\nu]} \left[ \left(\frac{\cos^2\phi}{n} \hat{\rho} \right)\oplus \left( \sin^2\phi\cdot\hat{M}_{\nu}\hat{\rho}\hat{M}_{\nu}^\dagger \right)\right] \left(\hat{1} \oplus \hat{0}\right)^\dagger = \hat{\rho}\oplus \hat{0}.
  \end{align}
\end{widetext}
Where the probability of successfully obtaining the outcome $\mu_0$ is given by:
\begin{align}
  P[\mu_0|\nu] = \frac{1}{P[\nu]}\frac{\cos^2\phi}{n}.
\end{align}
And since $\hat{\rho}\oplus\hat{0}$ is equivalent to $\hat{\rho}\otimes \ketbra{a}$,
by eliminating the ancillary system, we can recover the initial quantum state $\hat{\rho}$.

Since the probability of recovering the initial quantum state $(:= P[\text{rev}])$ corresponds to the probability of obtaining the outcome $\mu_0$ in this process, it is given by:
\begin{eqnarray}
  P[\text{rev}]:=P[\mu_0]= \sum_{\nu}P[\mu_0|\nu]P[\nu] = \cos^2\phi. \label{prob_rev}
\end{eqnarray}

Significantly, the recovery process $\mathcal{R}$ does not depend on the measurement outcome $\nu$, but the recovery process is constructed.

As shown above, we successfully constructed the quantum state recovery process utilizing the direct sum formalism without measurement outcomes.

\section{Discussions}

\subsection{How the Information Erasure Occurs} 
The key feature of quantum reversible measurement (QRM) is that any information obtained from the initial measurement is effectively erased by the subsequent inverse measurement~\cite{PhysRevA.83.044101,PhysRevLett.74.1040,buscemi2017information,Jordan_2010}.
Let us see how this information erasure occurs in our proposed scheme.

Before discussing the details of information erasure, we must first clarify the meaning of information in this context.
There are several approaches to quantifying the information gained in a measurement. In particular, in~\cite{PhysRevA.82.062306}, information is defined as the resource necessary to answer a question. In this context, the only available resource for answering arbitrary questions is the set of measurement outcomes obtained. More precisely, the extent to which these outcomes depend on the quantum state is critical in determining how much information is extracted during the measurement. Therefore, we examine how strongly the measurement outcomes depend on the initial quantum state.

For the initial measurement, the probability distribution is given by Eq.\eqref{P[nu]}:
\begin{align}
  P[\nu] = \frac{\cos^2\phi}{n} + \sin^2\phi \cdot \text{tr}\!\left( \hat{M}_{\nu}\hat{\rho}\hat{M}_{\nu}^\dagger \right).
\end{align}
If $\text{tr}\left( \hat{M}_{\nu}\hat{\rho}\hat{M}_{\nu}^\dagger \right)$ depends on the quantum state $\hat{\rho}$, this indicates that information about the state is indeed extracted by the measurement.

Next, consider the posterior probability of obtaining the outcome $\nu$ given that the reversal was successful, denoted by $P[\nu|\mu_0]$. By Bayes’ theorem, this is given by:
\begin{align}
  P[\nu|\mu_0] = \frac{P[\mu_0|\nu]P[\nu]}{P[\mu_0]} = \frac{1}{n}.
\end{align}
This result implies that, although the initial measurement may reveal some information about the quantum state, 
this information is completely erased if successful reversal is taken as a condition. Consequently, information erasure is an inherent part of the recovery process.

The following example can illustrate the non-triviality of this phenomenon. Suppose Alice performs steps (A) through (C), after which Bob carries out step (D). Once Alice has completed her task, she leaves the laboratory without communicating the measurement outcome to Bob. Nevertheless, if Bob subsequently obtains the outcome $\mu_0$ when measuring the recovered state $\hat{\rho}$, he is able to erase the information that Alice possessed—despite having no knowledge of her measurement result and without any communication with Alice. This challenges the conventional understanding, as in the case of QRM, that access to the measurement outcomes is necessary in order to erase the acquired information.

Moreover, estimating the original quantum state by performing a measurement followed by its reversal is impossible. This observation is consistent with the fundamental principle that any information gained through measurement inevitably disturbs the quantum state, as is also the case in QRM.

\subsection{Trade-off between Information Gain and Reversibility}
As a second aspect of discussion, we examine the trade-off between information gain and reversibility in comparison with QRM.

As we discussed, the information obtained from measurement is determined by the probability distribution of outcomes. To compare QRM with our scheme, we analyze these distributions under the same input state.

As shown in Sec.~\ref{sec:QRM}, QRM restores the original state probabilistically. it first measures specific conditions, then applies a second measurement containing the inverse of the first operator. For a direct comparison, we require the first measurement probabilities to coincide—denote those of QRM as $P^{\text{(QRM)}}[\nu]$—and then compare reversal probabilities. We will show that, under this condition, the reversibility of our scheme matches that of QRM. Thus, both exhibit the same trade-off between information gain and reversibility.

Suppose $\hat{\rho}$ lies in $\mathcal{H}_d^\perp$, and a measurement operator $\hat{m}_\nu$ acts on the same subspace. The condition for identical outcome probabilities is:
\begin{align}
\text{tr}(\hat{m}_{\nu} \hat{\rho} \hat{m}_{\nu}^\dagger) &= \frac{\cos^2 \phi}{n} + \sin^2 \phi \cdot \operatorname{tr}(\hat{M}_{\nu} \hat{\rho} \hat{M}_{\nu}^\dagger).
\end{align}
Since this holds for arbitrary quantum state $\hat{\rho}$, the measurement operator $\hat{m}_\nu$ satisfies the following relation:
\begin{align}
\hat{m}_{\nu}^\dagger \hat{m}_{\nu} 
= \frac{\cos^2 \phi}{n} + \sin^2 \phi \cdot \hat{M}_{\nu}^\dagger \hat{M}_{\nu}. \label{condition:comparison}
\end{align}
Applying singular value decomposition (SVD), let the singular values of $\hat{m}_\nu$ be $\lambda_i$, and let $\hat{U}$ be a unitary operator with $\{ \ket{\Psi_i^\perp} \}_i$ as an orthonormal basis. Then we have:
\begin{align}
\hat{m}_{\nu}^\dagger \hat{m}_{\nu} 
= \hat{U} \left( \sum_i \lambda_i^2 \ketbra{\Psi_i^\perp} \right) \hat{U}^\dagger.
\end{align}
Since $\hat{m}_\nu$ must be a valid measurement operator, we require $0 \leq \lambda_i^2$ for all $i$.
Now, solving for $\hat{M}_\nu$ from Eq.\eqref{condition:comparison}, we obtain:
\begin{align}
\hat{M}_{\nu}^\dagger \hat{M}_{\nu} 
= \hat{U} \left( \sum_i \frac{1}{\sin^2 \phi}\left( \lambda_i^2 - \frac{\cos^2 \phi}{n} \right) \ketbra{\Psi_i^\perp} \right) \hat{U}^\dagger.
\end{align}
In this expression, the terms $(\sin^2 \phi)^{-1}\left( \lambda_i^2 - {\cos^2 \phi}/{n} \right)$ can be interpreted as the singular values of $\hat{M}_{\nu}^\dagger \hat{M}_{\nu}$. To ensure positivity of $\hat{M}_{\nu}^\dagger \hat{M}_{\nu}$, the following condition must be satisfied:
\begin{align}
\frac{1}{\sin^2\phi}\left(\lambda_i^2 - \frac{\cos^2 \phi}{n}\right) \geq 0, \quad \forall i.
\end{align}
Therefore, the minimum value $\lambda_i$ can take is given as follows:
\begin{align}
 (\lambda_i^2)^{\text{min}} = \frac{\cos^2 \phi}{n}, \quad \forall i \label{lambda_phi}
\end{align}
On the other hand, as shown in~\cite{Terashima_2006}, the maximum probability of reversibility $P^{(\text{QRM})}[\mathrm{rev}]$ is also given by:
\begin{align}
P^{(\text{QRM})}[\mathrm{rev}] = \sum_\nu (\lambda_i^2)^{\text{min}}.
\end{align}
Thus, from Eq.\eqref{lambda_phi}, the maximum probability of reversibility $P^{(\text{QRM})}[\mathrm{rev}]$ is expressed using $\phi$ as follows: 
\begin{align}
  P^{(\text{QRM})}[\mathrm{rev}] = \cos^2 \phi.
\end{align}
Since $P[\mathrm{rev}]$ of our scheme is $\cos^2 \phi$ as given in Eq.\eqref{prob_rev}, the reversing probability is in agreement with $P^{(\text{QRM})}[\mathrm{rev}]$ of QRM as shown above.

Thus, by setting the measurement probabilities to be equal, the probability of successful restoration (i.e., reversibility) becomes identical for both methods.

Therefore, we conclude that the trade-off between information gain and reversibility is identical for our proposed scheme and the QRM. This finding also challenges the common intuition that utilizing measurement outcomes should allow for greater reversibility at a fixed level of information gain. As demonstrated above, this intuition does not hold.

\section{Conclusion}
In this work, we introduced a scheme for the probabilistic recovery of quantum states that dispenses with the need for measurement records by exploiting structural features of quantum systems. The central components are the notion of a quasi-copied state, which transfers the state into the orthogonal complement of the Hilbert space, and a special class of measurements defined thereon.

This direct-sum formulation, valid for arbitrary ancillary pure states, extends beyond conventional tensor-product approaches and enables a recovery protocol of both operational and informational significance. Operationally, it demonstrates the feasibility of recovery without outcome-dependent feedback, a previously regarded implausible scenario. 
Informationally, it shows that erasure of the information from a measurement can be achieved through the operation on the system, without acting on the register that stores the outcomes, which challenges the prevailing intuition regarding the necessity of memory access in informational erasure.
Furthermore, we have shown that the trade-off between information gain and reversibility is identical for the QRM and our proposed scheme. 
This finding contradicts the common belief that reversibility can be enhanced by conditioning on measurement outcomes, thereby underscoring the universality of this trade-off.

Finally, the independence from outcome records also simplifies the experimental realization of information erasure, removing the feedback control requirement. Taken together, these results highlight both the conceptual novelty—rooted in interference effects within the direct-sum structure—and the practical potential of our scheme for advancing the understanding of quantum measurement and information processing.

\section*{Acknowledgements}
We are grateful to Yuki Ito for his valuable comments.
This study was financially supported by programs for bridging the gap between R\&D and IDeal society (Society 5.0) and Generating Economic and social value (BRIDGE) and Cross-ministerial Strategic Innovation Promotion Program (SIP) from the Cabinet Office 23836436.

\bibliography{references}{}

\newpage
\appendix
\section{Theorems}\label{app:theorems}
We will show the corollaries that are used in this paper.
\begin{cor}[Operators on the direct sum space] \label{Product of operators on the direct sum space}
  The following properties hold for the product of operators on the direct sum space:
  \begin{align}
    \begin{split}
      (\hat{A}\oplus\hat{B})(\hat{C}\oplus\hat{D}) &= (\hat{A}\hat{C})\oplus(\hat{B}\hat{D}),\\
    (\hat{A}\oplus\hat{B})(\hat{C}\boxplus\hat{D}) &= (\hat{A}\hat{C})\boxplus(\hat{B}\hat{D}),\\
    (\hat{A}\boxplus\hat{B})(\hat{C}\oplus\hat{D}) &= (\hat{A}\hat{D})\boxplus(\hat{B}\hat{C}),\\
    (\hat{A}\boxplus\hat{B})(\hat{C}\boxplus\hat{D}) &= (\hat{A}\hat{D})\oplus(\hat{B}\hat{C}).
    \end{split}
  \end{align}
\end{cor}
\begin{proof}
  Immediate from the definitions of the direct sum operators.
\end{proof}

\begin{cor}[Sum of operators on the direct sum space] \label{Sum of operators on the direct sum space}
  The following properties hold for the sum of operators on the direct sum space:
  \begin{align}
    \begin{split}
    \hat{A}\oplus\hat{B} + \hat{C}\oplus\hat{D} &= (\hat{A}+\hat{C})\oplus(\hat{B}+\hat{D}),\\
    \hat{A}\boxplus\hat{B} + \hat{C}\boxplus\hat{D} &= (\hat{A}+\hat{C})\boxplus(\hat{B}+\hat{D}).
    \end{split}
  \end{align}
\end{cor}
\begin{proof}
  Immediate from the definitions of the direct sum operators.
\end{proof}

\begin{cor} \label{Conjucate of operators}
  The self-conjugate of direct sum operators is calculated as follows:
  \begin{align}
    \begin{split}
    (\hat{A}\oplus \hat{B})^\dagger = \hat{A}^\dagger \oplus \hat{B}^\dagger,\\
    (\hat{A}\boxplus \hat{B})^\dagger = \hat{B}^\dagger \boxplus \hat{A}^\dagger.
    \end{split}
  \end{align}
\end{cor}
\begin{proof}
  Let $\mathcal{H}_d$ and $\mathcal{H}_d^\bot$ be Hilbert spaces and consider the direct sum $\mathcal{H}_d \oplus \mathcal{H}_d^\bot$. Any vectors $\ket{\Psi}, \ket{\Phi}\in \mathcal{H}_d \oplus \mathcal{H}_d^\bot$ can be decomposed into only vectors $\ket{\psi}, \ket{\phi} \in \mathcal{H}_d$ and $\ket{\psi^\bot}, \ket{\phi^\bot} \in \mathcal{H}_d^\bot$, we have:
  \begin{align}
    \ket{\Psi} = \ket{\psi} + \ket{\psi^\bot},\quad \ket{\Phi} = \ket{\phi} + \ket{\phi^\bot}
  \end{align}
  Therefore, for the operator $\hat{A} \boxplus \hat{B}$ acting on the space $\mathcal{H}_d \oplus \mathcal{H}_d^\dagger$, 
\begin{align}
  \begin{split}
    &\braket{\Psi}{(\hat{A} \boxplus \hat{B})\Phi}\\
    &= \braket{(\psi + \psi^\bot)}{(\hat{A} \boxplus \hat{B})(\phi +\phi^\bot)} \\
    &= \braket{\psi^\bot}{\hat{B} \phi}+\braket{\psi}{\hat{A} \phi^\bot} \\
    &= \braket{\hat{B}^\dagger \psi^\bot}{\phi} + \braket{\hat{A}^\dagger \psi}{\phi^\bot}\\
    &= \braket{(\hat{B}^\dagger \psi^\bot) + (\hat{A}^\dagger \psi)}{\phi + \phi^\bot} \\
    &= \braket{(\hat{B}^\dagger \boxplus \hat{A}^\dagger)(\psi + \psi^\bot) } { \phi+ \phi^\bot}\\
    &= \braket{(\hat{B}^\dagger \boxplus \hat{A}^\dagger)\Psi}{\Phi}
    \end{split}
  \end{align}
 Therefore
  \begin{align}
    (\hat{A} \boxplus \hat{B})^\dagger = \hat{B}^\dagger \boxplus \hat{A}^\dagger.
  \end{align}
  Similarly, we can show 
  \begin{align}
    (\hat{A} \oplus \hat{B})^\dagger = \hat{A}^\dagger \oplus \hat{B}^\dagger.
  \end{align}
\end{proof}

\begin{cor} \label{Properties of traces in direct sum spaces}
  The following holds for traces of operators on the direct sum space, respectively:
  \begin{align}
    \text{tr}(\hat{A}\oplus\hat{B}) &= \text{tr}(\hat{A}) + \text{tr}(\hat{B}).
  \end{align}
\end{cor}
\begin{proof}
  Let's set the an orthonormal basis of $\mathcal{H}_d\oplus \mathcal{H}_d^\bot$:
  \begin{align}
    \{ \ket{e_1}, \ldots, \ket{e_d},  \ket{e^\bot_1}, \ldots, \ket{e_d^\bot}\}.
  \end{align}
  Therefore
  \begin{align}
    \begin{split}
    \text{tr}(\hat{A}\oplus\hat{B}) &= \sum_i \bra{e_i} (\hat{A}\oplus\hat{B}) \ket{e_i}+  \sum_i \bra{e_i^\bot } (\hat{A}\oplus\hat{B}) \ket{e_i^\bot}\\
    &= \sum_i \bra{e_i} \hat{A}  \ket{e_i} +  \sum_i \bra{e_i^\bot } \hat{B} \ket{e_i^\bot}\\
    &= \text{tr}(\hat{A}) + \text{tr}(\hat{B})
    \end{split}
  \end{align}
  This holds for any orthonormal basis. So the Corollary is proved.
\end{proof}
Using the corollaries we showed above, we can show the validity of the proposed measurement process in this paper.

\section{The validity of the measurement process}\label{app:validity}
Firstly, let us prove the validity of $\hat{K}_0$ and $\hat{K}_1$ defined by Eq.\eqref{K_0,K_1} as Kraus operators. Recall the operators defined by
\begin{align}
  \begin{split}
    \hat{K}_0 &= (\cos\phi \cdot \hat{1}) \oplus \hat{0} + (\sin\phi \cdot \hat{1}) \boxplus \hat{0}, \\
  \hat{K}_1 &= \hat{0} \oplus (-\cos\phi \cdot \hat{1}) + \hat{0} \boxplus (\sin\phi \cdot \hat{1}).
  \end{split}
\end{align}
Since for any linear operator $\hat{A}$ on a finite Hilbert space, the operator $\hat{A}^\dagger \hat{A}$ is a positive operator, the positivity condition is automatically satisfied.
As for the completeness condition, we observe that
\begin{align}
  \begin{split}
  \hat{K}_0^\dagger \hat{K}_0 + \hat{K}_1^\dagger \hat{K}_1 
  &= ( \cos^2 \phi + \sin^2 \phi ) \hat{1} \oplus ( \sin^2 \phi + \cos^2 \phi ) \hat{1} \\
  &= \hat{1} \oplus \hat{1} = \hat{1}.
  \end{split}
\end{align}

Therefore, $\{\hat{K}_0, \hat{K}_1\}$ form a valid set of Kraus operators: each $\hat{K}_i^\dagger \hat{K}_i$ is positive and their sum equals the identity operator on $\mathcal{H}_d \oplus \mathcal{H}_d^\perp$.

Next, we will show the validity of the measurement process defined by Eq.\eqref{measurement} and Eq.\eqref{recovery}.
Before directly going to the concrete form of the measurement operators, the following corollary stands immediately from the definitions of direct sum operators:
\begin{cor}[Sum of positive operators] \label{Sum of positive operators}
  If arbitrary linear operators $\hat{A}\in \mathcal{L}(\mathcal{H}_d)$ and $\hat{B}\in \mathcal{L}(\mathcal{H}_d^\bot)$ are the positive operators, $\hat{A}\oplus \hat{B}$ is also the positive operator.
\end{cor}
This corollary shows that positivity of the measurement process defined by Eq.\eqref{measurement} and Eq.\eqref{recovery} is valid.
Thus, we need to show that the measurement process satisfies the condition of completeness.

As for Eq.\eqref{measurement}, using Corollary~\ref{Sum of operators on the direct sum space} and Corollary~\ref{Conjucate of operators}, we can show that
\begin{align}
  \begin{split}
  \sum_{\nu} \left( \frac{1}{\sqrt{n}}\hat{1}\oplus\hat{M}_\nu \right)^\dagger\left( \frac{1}{\sqrt{n}}\hat{1}\oplus\hat{M}_\nu \right)&=\sum_{\nu} \left( \frac{1}{n}\hat{1}\oplus \hat{M}_\nu^\dagger \hat{M}_\nu \right)\\
  &= \hat{1}\oplus \hat{1} = \hat{1}.
  \end{split}
\end{align}
As for Eq.\eqref{recovery}, we can show that
\begin{align}
    \left( \hat{1}\oplus \hat{0} \right)^\dagger \left( \hat{1}\oplus \hat{0} \right) + \left( \hat{0}\oplus \hat{1} \right)^\dagger \left( \hat{0}\oplus \hat{1} \right)= \hat{1}\oplus \hat{1} = \hat{1}.
\end{align}
Therefore, since the measurement process defined by Eq.\eqref{measurement} and Eq.\eqref{recovery} satisfies the positivity and completeness conditions, these are valid measurement operators.

\section{The special case}\label{app:special_case}
In this section, we consider a special case of the proposed scheme.

This special case arises when the orthonormal basis of the ancillary qubit is used as the computational basis of the qubit in the quantum circuit—namely, the eigenstates of the Pauli-$Z$ operator $\hat{\sigma}_z$, denoted by ${\ket{0}, \ket{1}}$.

In this setting, the scheme admits a particularly transparent representation on a quantum circuit, as illustrated in Fig.~\ref{Diagram of this scheme}. The operations must be expressed in terms of tensor products to depict the scheme within a circuit framework.

In particular, for Step~\ref{step:measurement} of our scheme, which implements a measurement on the copied state using the measurement operators defined in Eq.\eqref{measurement}
\begin{align}
\left\{ \frac{1}{\sqrt{n}} \hat{1} \oplus \hat{M}_{\nu} \right\}_\nu,
\end{align}
the direct-sum form can be rewritten in tensor-product notation as
\begin{align}
\frac{1}{\sqrt{n}}\hat{1} \otimes \ketbra{0} + \hat{M}_{\nu} \otimes \ketbra{1}.
\end{align}
This corresponds to a conditional measurement: if the ancillary qubit is in the state $\ket{1}$, the measurement $\hat{M}_{\nu}$ is performed on the principal system; otherwise, the system remains unchanged. Consequently, the overall state evolves into a coherent superposition of a measured and an unmeasured state. This conditional structure is explicitly shown in the circuit diagram of Fig.\ref{Diagram of this scheme}.

Similarly, Step~\ref{step:recovery}, which stochastically recovers the original quantum state, is defined by the measurement operators in Eq.\eqref{recovery}:
\begin{align}
\{\hat{R}_{\mu_0} := \hat{1} \oplus \hat{0},\quad \hat{R}_{\mu_1} := \hat{0} \oplus \hat{1}\}.
\end{align}
Expressed in tensor-product form, these become
\begin{align}
\hat{R}_{\mu_0} = \hat{1} \otimes \ketbra{0},
\quad
\hat{R}_{\mu_1} = \hat{1} \otimes \ketbra{1}.
\end{align}
Thus, Step~\ref{step:recovery} corresponds to a projective measurement of the Pauli-$Z$ operator on the ancillary qubit:
\begin{align}
\hat{\sigma}_z := \ketbra{0} - \ketbra{1}.
\end{align}
In Fig.\ref{Diagram of this scheme}, this is represented as a diagram of a projective Pauli-$Z$ measurement on the ancillary qubit.
\begin{figure}[h]
  \centering
  \includegraphics{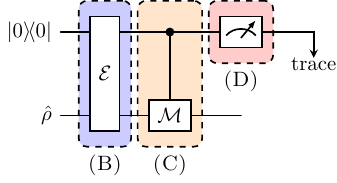}
  \caption{Diagram of the proposed scheme of a special case represented as a quantum circuit. Step~\ref{step:copying}, \ref{step:measurement}, and \ref{step:recovery} are highlighted.}
  \label{Diagram of this scheme}
\end{figure}

\end{document}